\documentclass[12pt,reqno]{article}

\usepackage[usenames]{color}
\usepackage{amssymb}
\usepackage{graphicx}
\usepackage{amscd}

\usepackage[colorlinks=true,
linkcolor=webgreen,
filecolor=webbrown,
citecolor=webgreen]{hyperref}

\definecolor{webgreen}{rgb}{0,.5,0}
\definecolor{webbrown}{rgb}{.6,0,0}

\usepackage{color}
\usepackage{fullpage}
\usepackage{float}

\usepackage{graphics,amsmath,amssymb}
\usepackage{amsthm}

\usepackage{tikz}
\usetikzlibrary{calc,intersections,through,backgrounds}

\theoremstyle{plain}
\newtheorem{theorem}{Theorem}
\newtheorem{corollary}[theorem]{Corollary}
\newtheorem{lemma}[theorem]{Lemma}
\newtheorem{proposition}[theorem]{Proposition}

\theoremstyle{definition}

\newtheorem{conjecture}[theorem]{Conjecture}
\newtheorem{problem}{Problem}

\theoremstyle{remark}
\newtheorem{remark}[theorem]{Remark}

\newcommand{\seqnum}[1]{\href{https://oeis.org/#1}{\underline{#1}}}
\begin{document}

\author{Anna E. Frid \\
Aix Marseille Univ, CNRS \\
Centrale Marseille, I2M \\ Marseille, France \\ \url{anna.e.frid@gmail.com}}
 
\title{Prefix palindromic length of the Thue-Morse word}
\date{~}

\maketitle

\begin{abstract}
The prefix palindromic length $PPL_u(n)$ of an infinite word $u$ is the minimal number of concatenated palindromes needed to express the prefix of length $n$ of $u$. In a 2013 paper with Puzynina and Zamboni we stated the conjecture that $PPL_u(n)$ is unbounded for every infinite word $u$ which
is not ultimately periodic. Up to now, the conjecture has been proven for almost all words, including all words avoiding some power $p$. However, even in that simple case the existing upper bound for the minimal number $n$ such that $PPL_u(n)>K$ is greater than any constant to the power $K$. Precise values of $PPL_u(n)$ are not known even for simplest examples like the Fibonacci word.

In this paper, we give the first example of such a precise computation and compute the function of the prefix palindromic length \seqnum{A307319} of the Thue-Morse word \seqnum{A010060}, a famous test object for all functions on infinite words. It happens that this sequence is $2$-regular, which raises the question if this fact can be generalized to all automatic sequences.
\end{abstract}

\section{Introduction}
By the usual definition, a palindrome is a finite word $p=p[1]\cdots p[n]$ on a finite alphabet such that $p[i]=p[n-i+1]$ for every $i$. We consider decompositions of a finite word $s$ to a minimal number of palindromes which we call a {\it palindromic length} of $s$: for example, the palindromic length of $abbaba$ is equal to 3 since this word is not a concatenation of two palindromes, but $abbaba=(abba)(b)(a)=(a)(bb)(aba)$. A decomposition to a minimal possible number of palindromes is called {\it optimal}.

In this paper, we are interested in the palindromic length of prefixes of an infinite word $u=u[1]\cdots u[n]\cdots$, denoted by $PPL_u(n)$. 
The length of the shortest prefix of $u$ of palindromic length $k$ is denoted by $SP_u(k)$ and can be considered as a kind of an inverse function to $PPL_u(n)$. Clearly, $SP_u(k)$ can be infinite: for example, if $u=abababab\cdots$, $SP_u(k)=\infty$ for every $k \geq 3$. 

The following conjecture was first formulated, in slightly different terms, in our 2013 paper with Puzynina and Zamboni \cite{fpz}.

\begin{conjecture}\label{c1}
 For every non ultimately periodic word $u$, the function $PPL_u(n)$ is unbounded, or, which is the same, $SP_u(k)<\infty$ for every $k \in \mathbb N$.
\end{conjecture}

In fact, there were two versions of the conjecture considered in our paper \cite{fpz}, one with the prefix palindromic length and the other with the palindromic length of any factor of $u$. However, Saarela \cite{saarela} later proved the equivalence of these two statements.

In the same initial paper \cite{fpz}, the conjecture was proved for the case when $u$ is $p$-power-free for some $p$, as well as for the more general case when a so-called $(p,l)$-condition holds for some $p$ and $l$. Due to the above-mentioned result by Saarela, this means that the conjecture is proven for almost all words, since almost all words contain as long $p$-power-free factors as needed. However, for some cases, the conjecture remains unsolved, and, for example, its proof for all Sturmian words \cite{frid} required a special technique. 

Most  published papers on palindromic length concern algorithmic aspects; in particular, there are several fast effective algorithms for computing $PPL_u(n)$ \cite{fici,shur2,shur1}.

The original proof of Conjecture \ref{c1} for the $p$-power-free words is not constructive. The upper bound for a length $N$ such that $PPL(N)\geq K$ for a given $K$ is given as a solution of a transcendental equation and grows with $K$ faster than any exponential function. However, this does not look the best possible bound. So, it is reasonable to state the following conjecture.

\begin{conjecture}\label{c2} If a word $u$ is $p$-power free for some $p$, then $$\lim \sup \frac{PPL_u(n)}{\ln n}>0,$$ or, which is the same, $SP_u(k)\leq C^k$ for some $C$. The constant $C$ can be chosen independently of $u$ as a function of $p$.
\end{conjecture}

In this paper, we consider in detail the case of the Thue-Morse word \seqnum{A010060}, a classical example of a word avoiding powers greater than 2 \cite{a_sh_ubi}. We give precise formulas for its prefix palindromic length and discuss its properties. This is a simple but necessary step before considering all $p$-power-free words, or all fixed points of uniform morphisms, or any other family of words containing the Thue-Morse word.

The results of this paper, in less detail, have been announced in the proceedings of DLT 2019 \cite{dlt}, together with some other results on the prefix palindromic length.

Throughout this paper, we use the notation $w(i..j]=w[i+1]..w[j]$ for a factor of a finite or infinite word $w$ starting at position $i+1$ and ending at $j$.

The following lemma is a particular case of a statement by Saarela \cite[L.\ 6]{saarela}. We give its proof for the sake of completeness.
\begin{lemma}\label{WRITEIT}
 For every word $u$ and for every $n\geq 0$, we have
 \[PPL_u(n)-1\leq PPL_u(n+1) \leq PPL_{u}(n)+1.\]
\end{lemma}
\begin{proof} Consider the prefixes $v$ and $va$ of $u$ of length $n$ and $n+1$ respectively. Clearly, for any decomposition $u=p_1\cdots p_k$ to $k$ palindromes $ua=p_1\cdots p_ka$ is a decomposition of $ua$ to $k+1$ palindrome. On the other hand, for any palindromic decomposition $ua=q_1\cdots q_k$, we have either $q_k=a$, and then $u=q_1\cdots q_{k-1}$, or $q_k=ap_k a$, for a (possibly empty) palindrome $p_k$, and then $u=q_1\cdots q_{k-1}ap_k$ is a decomposition of $u$ to $k+1$ palindromes. If initial decompositions were optimal, this gives $PPL_u(n+1)\leq PPL_u(n)+1$ and $PPL_u(n)\leq PPL_u(n+1)+1$. \end{proof}

So, the first differences of the prefix palindromic length can be equal only to -1, 0, or 1, and the graph never jumps.

In this paper, it is convenient to consider the famous Thue-Morse word \seqnum{A010060} 
\[t=abbabaabbaababba\cdots\]
as the fixed point starting with $a$ of the morphism
\[\tau: \begin{cases} a \to abba,\\ b \to baab. \end{cases}\]
Both images of letters under this morphism, which is the square of the usual Thue-Morse morphism $a \to ab, b \to ba$, are palindromes.

It is thus easy to see that every prefix of the Thue-Morse word of length $4^k$ is a palindrome, so that $PPL_t(4^k)=1$ for all $k\geq 0$. The first values of $PPL_t(n)$ and of $SP_t(k)$ are given below, see also the \seqnum{A307319} entry of the OEIS \cite{oeis}.
\medskip
\noindent
\begin{center}
\begin{tabular}{|c||c|c|c|c|c|c|c|c|c|c|c|c|c|c|c|c|}
  \hline
  $n$ &1 & 2& 3&4&5&6&7&8&9&10&11&12&13&14&15&16 \\
  \hline
  $PPL_t(n)$ & 1& 2& 2&1& 2& 3& 3&2& 3&4& 3& 2& 3& 3& 2& 1\\
  \hline
   
\end{tabular}
\end{center}

\medskip
As for the shortest prefix of a given palindromic length, we give its length in decimal and quaternary notation; see also the \seqnum{A320429} entry of the OEIS \cite{oeis}.

\medskip
\begin{center}
\begin{tabular}{|c||c|c|c|c|c|c|c|c|}
  \hline
  $k$ &1 & 2& 3&4&5&6&7&8 \\
  \hline
  $SP_t(k)$ & 1& 2& 6&10& 26& 90& 154&410\\
  \hline
  4-ary & 1&2&12&22&122&1122&2122&12122\\
  \hline
\end{tabular}
\end{center}

Now we are going to prove the self-similarity properties which we observe.
\section{Recurrence relations}

\begin{theorem}\label{t:tm}
The following identities hold for all $n\geq 0$:
\begin{align}
 PPL_t(4n)&=PPL_t(n), \label{e:4n}\\
PPL_t(4n+1)&=PPL_t(n)+1,\label{e:4n+1}\\
PPL_t(4n+2)&=\min(PPL_t(n),PPL_t(n+1))+2,\label{e:4n+2}\\
PPL_t(4n+3)&=PPL_t(n+1)+1.\label{e:4n+3}
\end{align}
\end{theorem}

\bigskip
To prove Theorem \ref{t:tm}, we need several observations. First of all, the shortest non-empty palindrome factors in the Thue-Morse word are $a$, $b$, $aa$, $bb$, $aba$, $bab$, $abba$, $baab$. All  palindromes of length more than 3 are of even length and have $aa$ or $bb$ in the center: if $t(i..i+2k]$ is a palindrome, then $t(i+k-1,i+k+1]=aa$ or $bb$. 

Let us say that an occurrence of a palindrome $t(i..j]$ is of type $(i', j')$ if $i'$ is the residue of $i$ and $j'$ is the residue of $j$ modulo 4. For example, the palindrome $t(5..7]=aa$ is of type $(1,3)$, the palindrome $t(4,8]=baab$ is of type $(0,0)$, and the palindrome $t(7..9]=bb$ is of type $(3,1)$.

\begin{proposition}\label{p:4-m}
Every occurrence of a palindromic factor of length not equal to one or three to the Thue-Morse word is of a type $(m,4-m)$ for some $m\in \{0,1,2,3\}$.
\end{proposition}
\begin{proof} Every such a palindrome in the Thue-Morse word is of even length which we denote by $2k$, and every occurrence of it is of the form $t(i..i+2k]$. Its center $t(i+k-1,i+k+1]$ is equal to $aa$ or $bb$, and these two words always appear in $t$ at positions of the form $t(2l-1,2l+1]$ for some $l \geq 1$. So, $i+k-1=2l-1$, meaning that $i=2l-k$ and $i+2k=2l+k$. So, modulo $4$, we have $i+(i+2k)=4l \equiv 0$, that is, $i \equiv -(i+2k)$. \end{proof}

\medskip
Note that the palindromes of odd length in the Thue-Morse word are, first, $a$ and $b$, which can be of type $(0,1)$, $(1,2)$, $(2,3)$ or $(3,0)$, and second, $aba$ and $bab$, which can only be of type $(2,1)$ or $(3,2)$.

\begin{proposition}\label{p:+-1}
 Let $t(i..i+k]$ for $i>0$ be a palindrome of length $k>0$ and of type $(m,4-m)$ for some $m \neq 0$. Then $t(i-1..i+k+1]$ is also a palindrome, as well as $t(i+1..i+k-1]$.
\end{proposition}
\begin{proof} The type of the palindrome is not $(0,0)$, meaning that its first and last letters $t[i+1]$ and $t[i+k]$ are not the first the last letters of $\tau$-images of letters. Since these  first and last letters are equal and their positions in $\tau$-images of letters are symmetric and determine their four-blocks $abba$ or $baab$, the letters $t[i]$ and $t[i+k+1]$ are also equal, and thus $t(i-1..i+k+1]$ is a palindrome. As for $t(i+1..i+k-1]$, it is a palindrome since is obtained from the palindrome $t(i..i+k]$ by erasing the first and the last letters. \end{proof}

\medskip
Let us say that a decomposition of $t(0..4n]$ to palindromes is a {\it $0$-decomposition} if all palindromes in it are of type $(0,0)$. The minimal number of palindromes in a 0-decomposition is denoted by $PPL^0_t(4n)$.
\begin{proposition}
 For every $n \geq 1$, we have $PPL_t(n)=PPL^0_t(4n)\geq PPL_t(4n)$.
\end{proposition}
\begin{proof} It is sufficient to note that $\tau$ is a bijection between all palindromic decompositions of $t(0..n]$ and 0-decompositions of $t(0..4n]$. \end{proof}

\begin{proposition}\label{p:+2+4}
 If \eqref{e:4n+2} holds for $n=N-1$, then 
 \begin{equation}\label{e:+2>+4}
PPL_t(4N-2)>PPL_t(4N).
\end{equation}
\end{proposition}
\begin{proof} The equality \eqref{e:4n+2} means that $PPL_t(4N-2)=\min(PPL_t(N-1),PPL_t(N))+2$, but since due to Lemma \ref{WRITEIT} we have $PPL_t(N)\leq PPL_t(N-1)+1$, we also have 
$\min(PPL_t(N-1),PPL_t(N))+2\geq  PPL_t(4N)+1$. \end{proof}

Now we can start the main proof of Theorem \ref{t:tm}.

The proof is done by induction on $n$. Clearly, $PPL_t(0)=0$, $PPL_t(1)=PPL_t(4)=1$, and $PPL_t(2)=PPL_t(3)=2$, the equalities \eqref{e:4n}--\eqref{e:4n+3} hold for $n=0$, and moreover, \eqref{e:4n} is true for $n=1$. Now suppose that they all, and, by Proposition \ref{p:+2+4}, the equality \eqref{e:+2>+4}, hold for all $n<N$, and  \eqref{e:4n} holds also for $n=N$. We fix an $N>0$ and prove for it the following sequence of propositions.

\begin{proposition}
 An optimal decomposition to palindromes of the prefix $t(0..4N+1]$ cannot end by a palindrome of length 3.
\end{proposition}
\begin{proof} Suppose the opposite: some optimal decomposition of $t(0..4N+1]$ ends by the palindrome $t(4N-2..4N+1]$. This palindrome is preceded by an optimal decomposition of $t(0..4N-2]$. So, $PPL_t(4N+1)=PPL_t(4N-2)+1$; but by \eqref{e:+2>+4} applied to $N-1$, which we can use by the induction hypothesis, $PPL_t(4N-2)>PPL_t(4N)$. So, $PPL_t(4N+1)>PPL_t(4N)+1$, contradicting to Lemma \ref{WRITEIT}. 
\end{proof}

\begin{proposition}
 There exists an optimal decomposition to palindromes of the prefix $t(0..4N+2]$ which does not end by a palindrome of length 3.
\end{proposition}
\begin{proof} The opposite would mean that all optimal decompositions of $t(0..4N+2]$ end by the palindrome $t(4N-1..4N+2]$ preceded by an optimal decomposition of $t(0..4N-1]$. So, $PPL_t(4N+2)=PPL_t(4N-1)+1$; by the induction hypothesis, $PPL_t(4N-1)=PPL_t(4N)+1$. So, $PPL_t(4N+2)=PPL_t(4N)+2$, and thus another optimal decomposition of $t(0..4N+2]$ can be obtained as an optimal decomposition of $t(0..4N]$ followed by two palindromes of length 1. A contradiction.
\end{proof}

\begin{proposition}\label{p:min123}
 For every $m\in \{1,2,3\}$, the equality holds
\[PPL_t(4N+m)=\min(PPL_t(4N+m-1),PPL_t(4N+m+1))+1.\]
\end{proposition}
\begin{proof} Consider an optimal decomposition $t(0..4N+m]=p_1\cdots p_k$, where $k=PPL_t(4N+m)$. Denote the ends of palindromes as $0=e_0<e_1<\cdots < e_k=4N+m$, so that $p_i=t(e_{i-1},e_i]$ for each $i$. Since $m\neq 0$ and due to Proposition \ref{p:4-m},  there exist some palindromes of length 1 or 3 in this decomposition. Let $p_j$ be the last of them. 

Suppose first that $j=k$. Then due to the two previous propositions, $p_k$ can be taken of length 1 not 3, so that  $t(0..4N+m-1]=p_1\cdots p_{k-1}$ is decomposable to $k-1$ palindromes. Due to Lemma \ref{WRITEIT}, we have $PPL_t(4N+m-1)=k-1$, and thus $PPL_t(4N+m)=PPL_t(4N+m-1)+1$. Again due to Lemma \ref{WRITEIT}, we have $PPL_t(4N+m+1)\geq PPL_t(4N+m)-1=PPL_t(4N+m-1)$, and so the statement holds. 

Now suppose that $j<k$, so that  $e_j\equiv -e_{j+1} \equiv e_{j+2} \equiv \cdots \equiv (-1)^{k-j} e_k \bmod 4$. Here $p_j$ is the last palindrome in an optimal decomposition of $p_1\cdots p_j$ and it is of length 1 or 3. But  if $e_j\equiv 1$ or $2 \bmod 4$, $p_j$ can be taken of length 1 due to the two previous propositions applied to some smaller length; and if $e_j\equiv 3 \bmod 4$, it is of length 1 since the suffix of length $3$ of $t(0..4n+3])$ is equal to $abb$ or to $baa$, so, it is not a palindrome. So, anyway, we can take $p_j$ of length one: $p_j=t(e_{j}-1,e_j]$. 

Since $e_j\equiv \pm e_k$ and $e_k \equiv m \neq 0 \bmod 4$, we may apply Proposition \ref{p:+-1} and see that $p_j'=t(e_j-1..e_{j+1}+1]$ is a palindrome, as well as $p_{j+1}'=t(e_{j+1}+1..e_{j+2}-1]$ and so on up to $p_{k-1}'=t(e_{k-1}+(-1)^{k-j}..e_k-(-1)^{k-j}]$. So, $p_1 \cdots p_{j-1} p_j' \cdots p_{k-1}'$ is a decomposition of $t(0..4N+m-(-1)^{k-j}]$ to $k-1$ palindromes. So, as above, $PPL_t(4N+m)=PPL_t(4N+m-(-1)^{k-j})+1$, and since $PPL_t(4N+m+(-1)^{k-j})\geq PPL_t(4N+m)-1=PPL_t(4N+m-(-1)^{k-j})$, the proposition holds. \end{proof}

\medskip

\begin{proposition}\label{p:4n}
Every optimal palindromic decomposition of $t(0..4N+4]$ is a 0-decomposition, and thus $PPL_t(4N+4)=PPL_t(N+1)$.
\end{proposition}
\begin{proof} Suppose the opposite; then the last palindrome in the optimal decomposition which is not of type (0,0) is of type $(m,0)$ and thus is of length 1 not 3. Since the proof of Theorem \ref{t:tm} proceeds by induction on $N$, this proposition is true for all $n<N$, and thus the palindrome of type $(m,0)$ is the very last palindrome of the optimal decomposition. Since the suffix of length $3$ of $t(0..4N+4]$ is equal to $bba$ or $aab$ and thus is not a palindrome,  the last palindrome of the optimal decomposition is of length $1$, meaning that $PPL_t(4N+4)=PPL_t(4N+3)+1$. Now let us use Proposition \ref{p:min123} applied to $m=3,2,1$; every time we get $PPL_t(4N+m)=PPL_t(4N+m-1)+1$. Summing up these inequalities, we get $PPL_t(4N+4)=PPL_t(4N)+4$, which is impossible since $PPL_t(4N)=PPL_t(N)$ and $PPL_t(4N+4)\leq PPL_t(N+1)\leq PPL_t(N)+1$. A contradiction. \end{proof}

\medskip
We have proven 
 \eqref{e:4n} for $n=N+1$. It remains to prove \eqref{e:4n+1}--\eqref{e:4n+3} for $n=N$. Indeed, we know that
\begin{equation}\label{e:-1+1}
-1 \leq PPL_t(4N+4)-PPL_t(4N)=PPL_t(N+1)-PPL_t(N)\leq 1.
\end{equation}
Now to prove \eqref{e:4n+1} suppose by contrary that $PPL_t(4N+1)\leq PPL_t(4N)=PPL_t(N)$. Due to Proposition \ref{p:min123}, this means that $PPL_t(4N+1)=PPL_t(4N+2)+1$, that is, $PPL_t(4N+2)<PPL_t(4N)$, and, again by Proposition \ref{p:min123}, $PPL_t(N+1)=PPL_t(4N+2)-2$. Thus, $PPL_t(N)-PPL_t(N+1)\geq 3$, a contradiction to \eqref{e:-1+1}.
So, \eqref{e:4n+1} is proven. 

The equality \eqref{e:4n+3} is proven symmetrically. Now \eqref{e:4n+2} follows from both these equalities in combination with Proposition \ref{p:min123}, completing the proof of Theorem \ref{t:tm}.

\section{Corollaries}

The first differences $(d_t(n))_{n=0}^{\infty}$ of the prefix palindromic length are defined as  $d_t(n)=PPL_t(n+1)-PPL_t(n)$; here we set $PPL_t(0)=0$. Due to Lemma \ref{WRITEIT}, $d_t(n)\in \{-1,0,+1\}$ for every $n$; so, it is a sequence on a finite alphabet which we prefer to denote $\{
 \begin{tikzpicture}[baseline=1pt,scale=0.3]
    \draw[thick] (0,1) -- (1,0);
 \node[circle,fill=black,inner sep=0pt,minimum size=4pt] (a) at (0, 1) {};
 \node[circle,fill=black,inner sep=0pt,minimum size=4pt] (a) at (1, 0) {};
  \end{tikzpicture},
\begin{tikzpicture}[baseline=-3pt,scale=0.3]
    \draw[thick] (0,0) -- (1,0);
 \node[circle,fill=black,inner sep=0pt,minimum size=4pt] (a) at (0, 0) {};
 \node[circle,fill=black,inner sep=0pt,minimum size=4pt] (a) at (1, 0) {};
  \end{tikzpicture},
\begin{tikzpicture}[baseline=1pt,scale=0.3]
     \draw[thick] (0,0) -- (1,1);
 \node[circle,fill=black,inner sep=0pt,minimum size=4pt] (a) at (0, 0) {};
 \node[circle,fill=black,inner sep=0pt,minimum size=4pt] (a) at (1, 1) {};
  \end{tikzpicture}
   \}$. We write these symbols joining the ends of intervals from left to right, so that the sequence $(d_t(n))$ becomes the plot of $PPL_t(n)$. 

The following corollary of Theorem \ref{t:tm} is more or less straightforward.
\begin{corollary}
 The sequence $(d_t(n))$ is the fixed point of the morphism
\[\delta: \begin{cases}
           
\begin{tikzpicture}[baseline=1pt,scale=0.3]
     \draw[thick] (0,0) -- (1,1);
 \node[circle,fill=black,inner sep=0pt,minimum size=4pt] (a) at (0, 0) {};
 \node[circle,fill=black,inner sep=0pt,minimum size=4pt] (a) at (1, 1) {};
  \end{tikzpicture}
  
  \to 
  
   \begin{tikzpicture}[baseline=1pt,scale=0.3]
     \draw[thick] (0,0) -- (1,1)-- (2,2)--(3,2)--(4,1);
 \node[circle,fill=black,inner sep=0pt,minimum size=4pt] (a) at (0, 0) {};
 \node[circle,fill=black,inner sep=0pt,minimum size=4pt] (a) at (1, 1) {};
 \node[circle,fill=black,inner sep=0pt,minimum size=4pt] (a) at (2, 2) {};
 \node[circle,fill=black,inner sep=0pt,minimum size=4pt] (a) at (3, 2) {};
 \node[circle,fill=black,inner sep=0pt,minimum size=4pt] (a) at (4, 1) {};
 
  \end{tikzpicture}
\\

\begin{tikzpicture}[baseline=1pt,scale=0.3]
    \draw[thick] (0,0) -- (1,0);
 \node[circle,fill=black,inner sep=0pt,minimum size=4pt] (a) at (0, 0) {};
 \node[circle,fill=black,inner sep=0pt,minimum size=4pt] (a) at (1, 0) {};
  \end{tikzpicture}
  
  \to 
  
   \begin{tikzpicture}[baseline=2pt,scale=0.3]
     \draw[thick] (0,0) -- (1,1)-- (2,2)--(3,1)--(4,0);
 \node[circle,fill=black,inner sep=0pt,minimum size=4pt] (a) at (0, 0) {};
 \node[circle,fill=black,inner sep=0pt,minimum size=4pt] (a) at (1, 1) {};
 \node[circle,fill=black,inner sep=0pt,minimum size=4pt] (a) at (2, 2) {};
 \node[circle,fill=black,inner sep=0pt,minimum size=4pt] (a) at (3, 1) {};
 \node[circle,fill=black,inner sep=0pt,minimum size=4pt] (a) at (4, 0) {};
   \end{tikzpicture}\\
   
 \begin{tikzpicture}[baseline=1pt,scale=0.3]
    \draw[thick] (0,1) -- (1,0);
 \node[circle,fill=black,inner sep=0pt,minimum size=4pt] (a) at (0, 1) {};
 \node[circle,fill=black,inner sep=0pt,minimum size=4pt] (a) at (1, 0) {};
  \end{tikzpicture}
  
  \to 
  
   \begin{tikzpicture}[baseline=1pt,scale=0.3]
     \draw[thick] (0,1) -- (1,2)-- (2,2)--(3,1)--(4,0);
 \node[circle,fill=black,inner sep=0pt,minimum size=4pt] (a) at (0, 1) {};
 \node[circle,fill=black,inner sep=0pt,minimum size=4pt] (a) at (1, 2) {};
 \node[circle,fill=black,inner sep=0pt,minimum size=4pt] (a) at (2, 2) {};
 \node[circle,fill=black,inner sep=0pt,minimum size=4pt] (a) at (3, 1) {};
 \node[circle,fill=black,inner sep=0pt,minimum size=4pt] (a) at (4, 0) {};
   \end{tikzpicture}  
          \end{cases}\]

\end{corollary}
{\sc Proof.} Theorem \ref{t:tm} immediately means that $PPL_t(4n),\ldots,PPL_t(4n+4)$ are determined by $PPL_t(n)$ and $PPL_t(n+1)$, and moreover, $d_t(4n)$, $\ldots$, $d_t(4n+3)$ are determined by $d_t(n)$. This means exactly that the sequence $d_t(n)$ is a fixed point of a morphism of length 4. The equality \eqref{e:4n+1} means that the first symbol of any morphic image of $\delta$ is $\begin{tikzpicture}[baseline=1pt,scale=0.3]
     \draw[thick] (0,0) -- (1,1);
 \node[circle,fill=black,inner sep=0pt,minimum size=4pt] (a) at (0, 0) {};
 \node[circle,fill=black,inner sep=0pt,minimum size=4pt] (a) at (1, 1) {};
  \end{tikzpicture}$; the equality \eqref{e:4n+3} means that the last symbol of any morphic image of $\delta$ is $ \begin{tikzpicture}[baseline=1pt,scale=0.3]
    \draw[thick] (0,1) -- (1,0);
 \node[circle,fill=black,inner sep=0pt,minimum size=4pt] (a) at (0, 1) {};
 \node[circle,fill=black,inner sep=0pt,minimum size=4pt] (a) at (1, 0) {};
  \end{tikzpicture} $; the two symbols in the middle can be found from \eqref{e:4n+2} and depend on $d_t(n)$. \hfill $\Box$

  \medskip
With the previous corollary, we can  draw the plot of $PPL_t(n)$ as the fixed point of $\delta$.

\begin{figure}[h]
\centering
   \begin{tikzpicture}[baseline=1pt,scale=0.1]
     \draw [<->] (0,7) node (yaxis) [above] {$PPL_t(n)$}
        |- (105,0) node (xaxis) [right] {$n$};
    \foreach \x/\xtext in {4/4, 16/16, 64/64}
    \draw[shift={(\x,0)}] (0pt,5pt) -- (0pt,-5pt) node[below] {\tiny $\xtext$};
     \draw[thick] (0,0)--(1,1)--(2,2)--(3,2)--(4,1)--(5,2)--(6,3)--(7,3)--(8,2)--(9,3)--(10,4)--(11,3)--(12,2)--(13,3)--(14,3)--(15,2)--(16,1)--(17,2)--(18,3)--(19,3)--(20,2)--(21,3)--(22,4)--(23,4)--(24,3)--(25,4)--(26,5)--(27,4)--(28,3)--(29,4)--(30,4)--(31,3)--(32,2)--(33,3)--(34,4)--(35,4)--(36,3)--(37,4)--(38,5)--(39,5)--(40,4)--(41,5)--(42,5)--(43,4)--(44,3)--(45,4)--(46,4)--(47,3)--(48,2)--(49,3)--(50,4)--(51,4)--(52,3)--(53,4)--(54,5)--(55,4)--(56,3)--(57,4)--(58,4)--(59,3)--(60,2)--(61,3)--(62,3)--(63,2)--(64,1)--(65,2)--(66,3)--(67,3)--(68,2)--(69,3)--(70,4)--(71,4)--(72,3)--(73,4)--(74,5)--(75,4)--(76,3)--(77,4)--(78,4)--(79,3)--(80,2)--(81,3)--(82,4)--(83,4)--(84,3)--(85,4)--(86,5)--(87,5)--(88,4)--(89,5)--(90,6)--(91,5)--(92,4)--(93,5)--(94,5)--(95,4)--(96,3)--(97,4)--(98,5)--(99,5)--(100,4);
 \node[circle,fill=black,inner sep=0pt,minimum size=2pt] (a) at (0,0) {};
\node[circle,fill=black,inner sep=0pt,minimum size=2pt] (a) at (1,1) {};
\node[circle,fill=black,inner sep=0pt,minimum size=2pt] (a) at (2,2) {};
\node[circle,fill=black,inner sep=0pt,minimum size=2pt] (a) at (3,2) {};
\node[circle,fill=black,inner sep=0pt,minimum size=2pt] (a) at (4,1) {};
\node[circle,fill=black,inner sep=0pt,minimum size=2pt] (a) at (5,2) {};
\node[circle,fill=black,inner sep=0pt,minimum size=2pt] (a) at (6,3) {};
\node[circle,fill=black,inner sep=0pt,minimum size=2pt] (a) at (7,3) {};
\node[circle,fill=black,inner sep=0pt,minimum size=2pt] (a) at (8,2) {};
\node[circle,fill=black,inner sep=0pt,minimum size=2pt] (a) at (9,3) {};
\node[circle,fill=black,inner sep=0pt,minimum size=2pt] (a) at (10,4) {};
\node[circle,fill=black,inner sep=0pt,minimum size=2pt] (a) at (11,3) {};
\node[circle,fill=black,inner sep=0pt,minimum size=2pt] (a) at (12,2) {};
\node[circle,fill=black,inner sep=0pt,minimum size=2pt] (a) at (13,3) {};
\node[circle,fill=black,inner sep=0pt,minimum size=2pt] (a) at (14,3) {};
\node[circle,fill=black,inner sep=0pt,minimum size=2pt] (a) at (15,2) {};
\node[circle,fill=black,inner sep=0pt,minimum size=2pt] (a) at (16,1) {};
\node[circle,fill=black,inner sep=0pt,minimum size=2pt] (a) at (17,2) {};
\node[circle,fill=black,inner sep=0pt,minimum size=2pt] (a) at (18,3) {};
\node[circle,fill=black,inner sep=0pt,minimum size=2pt] (a) at (19,3) {};
\node[circle,fill=black,inner sep=0pt,minimum size=2pt] (a) at (20,2) {};
\node[circle,fill=black,inner sep=0pt,minimum size=2pt] (a) at (21,3) {};
\node[circle,fill=black,inner sep=0pt,minimum size=2pt] (a) at (22,4) {};
\node[circle,fill=black,inner sep=0pt,minimum size=2pt] (a) at (23,4) {};
\node[circle,fill=black,inner sep=0pt,minimum size=2pt] (a) at (24,3) {};
\node[circle,fill=black,inner sep=0pt,minimum size=2pt] (a) at (25,4) {};
\node[circle,fill=black,inner sep=0pt,minimum size=2pt] (a) at (26,5) {};
\node[circle,fill=black,inner sep=0pt,minimum size=2pt] (a) at (27,4) {};
\node[circle,fill=black,inner sep=0pt,minimum size=2pt] (a) at (28,3) {};
\node[circle,fill=black,inner sep=0pt,minimum size=2pt] (a) at (29,4) {};
\node[circle,fill=black,inner sep=0pt,minimum size=2pt] (a) at (30,4) {};
\node[circle,fill=black,inner sep=0pt,minimum size=2pt] (a) at (31,3) {};
\node[circle,fill=black,inner sep=0pt,minimum size=2pt] (a) at (32,2) {};
\node[circle,fill=black,inner sep=0pt,minimum size=2pt] (a) at (33,3) {};
\node[circle,fill=black,inner sep=0pt,minimum size=2pt] (a) at (34,4) {};
\node[circle,fill=black,inner sep=0pt,minimum size=2pt] (a) at (35,4) {};
\node[circle,fill=black,inner sep=0pt,minimum size=2pt] (a) at (36,3) {};
\node[circle,fill=black,inner sep=0pt,minimum size=2pt] (a) at (37,4) {};
\node[circle,fill=black,inner sep=0pt,minimum size=2pt] (a) at (38,5) {};
\node[circle,fill=black,inner sep=0pt,minimum size=2pt] (a) at (39,5) {};
\node[circle,fill=black,inner sep=0pt,minimum size=2pt] (a) at (40,4) {};
\node[circle,fill=black,inner sep=0pt,minimum size=2pt] (a) at (41,5) {};
\node[circle,fill=black,inner sep=0pt,minimum size=2pt] (a) at (42,5) {};
\node[circle,fill=black,inner sep=0pt,minimum size=2pt] (a) at (43,4) {};
\node[circle,fill=black,inner sep=0pt,minimum size=2pt] (a) at (44,3) {};
\node[circle,fill=black,inner sep=0pt,minimum size=2pt] (a) at (45,4) {};
\node[circle,fill=black,inner sep=0pt,minimum size=2pt] (a) at (46,4) {};
\node[circle,fill=black,inner sep=0pt,minimum size=2pt] (a) at (47,3) {};
\node[circle,fill=black,inner sep=0pt,minimum size=2pt] (a) at (48,2) {};
\node[circle,fill=black,inner sep=0pt,minimum size=2pt] (a) at (49,3) {};
\node[circle,fill=black,inner sep=0pt,minimum size=2pt] (a) at (50,4) {};
\node[circle,fill=black,inner sep=0pt,minimum size=2pt] (a) at (51,4) {};
\node[circle,fill=black,inner sep=0pt,minimum size=2pt] (a) at (52,3) {};
\node[circle,fill=black,inner sep=0pt,minimum size=2pt] (a) at (53,4) {};
\node[circle,fill=black,inner sep=0pt,minimum size=2pt] (a) at (54,5) {};
\node[circle,fill=black,inner sep=0pt,minimum size=2pt] (a) at (55,4) {};
\node[circle,fill=black,inner sep=0pt,minimum size=2pt] (a) at (56,3) {};
\node[circle,fill=black,inner sep=0pt,minimum size=2pt] (a) at (57,4) {};
\node[circle,fill=black,inner sep=0pt,minimum size=2pt] (a) at (58,4) {};
\node[circle,fill=black,inner sep=0pt,minimum size=2pt] (a) at (59,3) {};
\node[circle,fill=black,inner sep=0pt,minimum size=2pt] (a) at (60,2) {};
\node[circle,fill=black,inner sep=0pt,minimum size=2pt] (a) at (61,3) {};
\node[circle,fill=black,inner sep=0pt,minimum size=2pt] (a) at (62,3) {};
\node[circle,fill=black,inner sep=0pt,minimum size=2pt] (a) at (63,2) {};
\node[circle,fill=black,inner sep=0pt,minimum size=2pt] (a) at (64,1) {};
\node[circle,fill=black,inner sep=0pt,minimum size=2pt] (a) at (65,2) {};
\node[circle,fill=black,inner sep=0pt,minimum size=2pt] (a) at (66,3) {};
\node[circle,fill=black,inner sep=0pt,minimum size=2pt] (a) at (67,3) {};
\node[circle,fill=black,inner sep=0pt,minimum size=2pt] (a) at (68,2) {};
\node[circle,fill=black,inner sep=0pt,minimum size=2pt] (a) at (69,3) {};
\node[circle,fill=black,inner sep=0pt,minimum size=2pt] (a) at (70,4) {};
\node[circle,fill=black,inner sep=0pt,minimum size=2pt] (a) at (71,4) {};
\node[circle,fill=black,inner sep=0pt,minimum size=2pt] (a) at (72,3) {};
\node[circle,fill=black,inner sep=0pt,minimum size=2pt] (a) at (73,4) {};
\node[circle,fill=black,inner sep=0pt,minimum size=2pt] (a) at (74,5) {};
\node[circle,fill=black,inner sep=0pt,minimum size=2pt] (a) at (75,4) {};
\node[circle,fill=black,inner sep=0pt,minimum size=2pt] (a) at (76,3) {};
\node[circle,fill=black,inner sep=0pt,minimum size=2pt] (a) at (77,4) {};
\node[circle,fill=black,inner sep=0pt,minimum size=2pt] (a) at (78,4) {};
\node[circle,fill=black,inner sep=0pt,minimum size=2pt] (a) at (79,3) {};
\node[circle,fill=black,inner sep=0pt,minimum size=2pt] (a) at (80,2) {};
\node[circle,fill=black,inner sep=0pt,minimum size=2pt] (a) at (81,3) {};
\node[circle,fill=black,inner sep=0pt,minimum size=2pt] (a) at (82,4) {};
\node[circle,fill=black,inner sep=0pt,minimum size=2pt] (a) at (83,4) {};
\node[circle,fill=black,inner sep=0pt,minimum size=2pt] (a) at (84,3) {};
\node[circle,fill=black,inner sep=0pt,minimum size=2pt] (a) at (85,4) {};
\node[circle,fill=black,inner sep=0pt,minimum size=2pt] (a) at (86,5) {};
\node[circle,fill=black,inner sep=0pt,minimum size=2pt] (a) at (87,5) {};
\node[circle,fill=black,inner sep=0pt,minimum size=2pt] (a) at (88,4) {};
\node[circle,fill=black,inner sep=0pt,minimum size=2pt] (a) at (89,5) {};
\node[circle,fill=black,inner sep=0pt,minimum size=2pt] (a) at (90,6) {};
\node[circle,fill=black,inner sep=0pt,minimum size=2pt] (a) at (91,5) {};
\node[circle,fill=black,inner sep=0pt,minimum size=2pt] (a) at (92,4) {};
\node[circle,fill=black,inner sep=0pt,minimum size=2pt] (a) at (93,5) {};
\node[circle,fill=black,inner sep=0pt,minimum size=2pt] (a) at (94,5) {};
\node[circle,fill=black,inner sep=0pt,minimum size=2pt] (a) at (95,4) {};
\node[circle,fill=black,inner sep=0pt,minimum size=2pt] (a) at (96,3) {};
\node[circle,fill=black,inner sep=0pt,minimum size=2pt] (a) at (97,4) {};
\node[circle,fill=black,inner sep=0pt,minimum size=2pt] (a) at (98,5) {};
\node[circle,fill=black,inner sep=0pt,minimum size=2pt] (a) at (99,5) {};

   \end{tikzpicture}

\caption{$PPL_t(n)$}\label{f:plot}
\end{figure}

The next proposition can be obtained from Theorem \ref{t:tm} by elementary computations. Recall that $SP(k)=SP_t(k)$ is the length $n$ of the shortest prefix of $t$ such that its palindromic length $PPL_t(n)$ is equal to $k$.

\begin{proposition}\label{p:sp}
We have $SP_t(1)=1$, $SP_t(2)=2$, $SP_t(3)=6$ and for all $k>0$,
\[SP_t(k+3)=16 SP_t(k)-6.\]
\end{proposition}
{\sc Proof.} Let us introduce $SP_2(k)$ as the minimal number $n$ such that $PPL_t(n)=PPL_t(n+1)=k$. By definition, $SP_2(k)\geq SP(k)$. The first values of $SP(k)$ and $SP_2(k)$ are given below.

\medskip
\begin{center}
\begin{tabular}{|c||c|c|c|c|c|}
  \hline
  $k$ &1 & 2& 3&4&5 \\
  \hline
  $SP(k)$ & 1& 2& 6&10& 26\\
  \hline
  $SP_2(k)$ & $\infty$ &2&6&22&38\\
  \hline
\end{tabular}
\end{center}

\medskip
From the definition of the morphism $\delta$ we immediately see that a new value $n=SP(k)$ can appear either in the middle of the $\delta$-image of $d(n')=d(SP_2(k-2))$, or in the middle of the $\delta$-image of $d(n'')$, where $n''=SP(k-1)-1$. The latter case is also the only possible way to get a new value $n=SP_2(k)$. So, 
\begin{equation}\label{e:min}
 SP(k)=\min(4SP_2(k-2)+2, 4SP(k-1)-2),
\end{equation}
\begin{equation}\label{e:sp2}
SP_2(k)=4SP(k-1)-2.
\end{equation}
As we see from the table, for $3\leq k\leq 5$, we have 
$SP(k-1)\leq SP_2(k-1)<SP(k)$. The first inequality is obvious, but let us prove the second one by induction. Its base is observed for $3\leq k\leq 5$, so, consider $k\geq 6$ such that for all $k'=k-3,k-2,k-1$ we have $SP_2(k'-1)<SP(k')$.

In particular, $SP_2(k-4)<SP(k-3)$, so due to \eqref{e:min}, we have $SP(k-2)=4SP_2(k-4)+2$, and so due to \eqref{e:sp2}, 
\begin{equation}\label{e:sp22}
 SP_2(k-1)=16SP_2(k-4)+6.
\end{equation}
 On the other hand, we have $SP_2(k-2)<SP(k-1)$, so, \eqref{e:min} becomes $SP(k)=4SP_2(k-2)+2$, and together with \eqref{e:sp2} this gives
 \begin{equation}\label{e:spk}
  SP(k)=16 SP(k-3)-6.
 \end{equation}
Combining \eqref{e:sp22}, \eqref{e:spk}, the induction base $SP_2(k-4)<SP(k-3)$ and the fact that all the values are integers, we obtain that $SP_2(k-1)<SP(k)$ for all $k\geq 3$. We also see that \eqref{e:spk} is true for all $k\geq 4$, proving this proposition. \hfill $\Box$

\medskip
The following corollary of the previous proposition can be proved by straightforward induction.

\begin{corollary}
In the 4-ary numeration system, we have  $SP(3k+2)=((12)^k2)_4$ for all $k \geq 0$; $SP(3k)=(1(12)^{k-1}2)_4$ for all $k \geq  1$; $SP(3k+1)=(2(12)^{k-1}2)_4$ for all $k \geq 1$.
\end{corollary}

Another direct consequence of Proposition \ref{p:sp} is
\begin{corollary}
 We have $$\displaystyle \lim \sup \frac{PPL_t(n)}{\ln n}=\frac{3}{4 \ln 2},$$ whereas $\displaystyle \lim \inf \frac{PPL_t(n)}{\ln n}=0$ since $PPL_t(4^m)=1$ for all $m$.
\end{corollary}

\section{Regularity}
The sequence $(PPL_t(n))$ is closely related to the Thue-Morse word, the most classical example of a 2-automatic sequence. In general, a sequence $w=(w[n])$ is called $k$-automatic if there exists a finite automaton such that for the input equal to the $k$-ary representation of $n$, the output is equal to $w[n]$. Equivalently, due to a theorem by Cobham \cite{cobham72}, a sequence is $k$-automatic if and only if it is an image under a coding $c: \Sigma \to \Delta$ of a fixed point of a $k$-uniform morphism $\varphi$: $w=c(w')$, where $w'=\varphi(w')$ \cite[Ch.\ 6]{a_sh}. So, the Thue-Morse word is 2-automatic since it is a fixed point of the $2$-uniform morphism $a \to ab, b \to ba$, and the sequence $(d(n))$ is 4-automatic since it is a fixed point of $\delta$. In both cases, the coding can be taken to be trivial: $c(x)=x$ for every letter $x$. It is also well-known that a sequence is $k$-automatic if and only if it is $k^m$-automatic for any integer $m$, so, the Thue-Morse word is also 4-automatic and the sequence $(d(n))$ is 2-automatic.

A more general notion of a {\it $k$-regular sequence} was introduced by Allouche and Shallit \cite{a_sh_kreg}, see also  \cite[Ch.\ 16]{a_sh}. A sequence $(a(n))$ is called $k$-regular (on $\mathbb Z$) if there exists a finite number of sequences $\{(a_1(n)),\ldots,(a_s(n))\}$  such that for every integer $i \geq 0$ and $0\leq b <k^i$ there exist $c_1,\ldots,c_s \in \mathbb Z$ such that for all $n\geq 0$ we have

\[a(k^i n +b)=\sum_{1\leq j \leq s} c_j a_j(n).\]

It is also known that a sequence is $k$-automatic if and only if it is $k$-regular and takes on finitely many values  \cite[Thm.\ 16.1.4]
{a_sh}.  Moreover, a sequence $a=(a(n))$ is $k$-regular if and only if there exist $r$ sequences $a_1=a,\ldots,a_r$ and a matrix-valued morphism $\mu$ such that if
\[V(n)=\begin{pmatrix} a_1(n) \\ a_2(n) \\ \cdots \\ a_r(n) \end{pmatrix},\]
then
\[V(kn+b)=\mu(b)V(n)\]
for $0\leq b < k$ \cite[Thm.\ 16.1.3]{a_sh}.
Many sequences related to $k$-automatic words are $k$-regular, as it was shown by Charlier, Rampersad and Shallit \cite{crs}. However, it seems that the general approach from the mentioned paper does not directly work for the palindromic length, so, we have to prove the following corollary only for the Thue-Morse word.

\begin{corollary}
 The sequence $PPL_t(n)$ is 4-regular.
\end{corollary}
{\sc Proof.} The sequence $(d(n))$, here considered on the alphabet $\{-1,0,1\}$, is 4-automatic as a fixed point of $\delta$ and thus 4-regular, as well as its image $(e(n))$ under the following coding: 
$$e(n)=\begin{cases} 1, d(n)=-1,\\ 0, \mbox{~otherwise}.\end{cases} $$
So, for each $b=0,1,2,3$ it is sufficient to combine the respective matrices $\mu$ for $d$ and $e$ with the following equalities equivalent to Theorem \ref{t:tm}:
\begin{align*}
 PPL_t(4n)&=PPL_t(n),\\
 PPL_t(4n+1)&=PPL_t(n)+1,\\
 PPL_t(4n+2)&=PPL_t(n)+2-e(n),\\
 PPL_t(4n+3)&=PPL_t(n)+d(n)+1. ~~~~~~~~\hfill ~ \Box
\end{align*}

\begin{remark}{\rm
 In fact, every sequence with $k$-automatic first differences is $k$-regular, which can be proven with a similar construction, perhaps with several auxiliary sequences like $(e(n))$.
 }
\end{remark}

\section{Conclusion}
Up to my knowledge, the  results of this paper are thus far only precise formulas for the prefix palindromic length of a non-trivial infinite word not constructed especially for that. 
Even for famous and simple examples like Toeplitz words or the Fibonacci word \seqnum{A003849}, lower bounds for the prefix palindromic length are difficult \cite{f:num,dlt}. The only more or less universal lower bounds for all $p$-power-free words are those from the first paper on the subject \cite{fpz}, with $SP(k)$ growing faster than any constant to the power $k$. Later \cite{f:num}, some calculations allowed a reasonable exponential conjecture on the $SP(k)$ of the Fibonacci word, but it is not clear how to prove it. So, the following more particular open questions can be added to general Conjectures \ref{c1} and \ref{c2}.

\begin{problem}
 Find a precise formula  for the prefix palindromic length of the period-doubling word \seqnum{A096268}, or a lower bound for its $\lim \sup$.
\end{problem}

\begin{problem}
 Find a precise formula  for the prefix palindromic length of the Fibonacci word \seqnum{A003849}, or a lower bound for its $\lim \sup$.\end{problem}
 
 \begin{problem}
  Is it true that the function $PPL_u(n)$ is $k$-regular for any $k$-automatic word $u$? Fibonacci-regular for the Fibonacci word?
 \end{problem}

\end{document}